\documentclass[12pt]{article} 
\usepackage{url}
\usepackage[colorlinks, citecolor={blue}]{hyperref}
\usepackage{amsfonts,amscd,amssymb}
\usepackage{amsthm,amsmath,natbib}
\usepackage{algorithmic,algorithm}
\usepackage{bm}
\usepackage{bbm} 
\usepackage{color}
\usepackage{verbatim}
\usepackage{graphicx}
\usepackage{setspace}
\usepackage{natbib}
\usepackage[margin=1in]{geometry}
\usepackage{enumitem}
\usepackage{enumerate}
\usepackage[textsize=tiny]{todonotes}
\usepackage[toc,page]{appendix}
\usepackage[mathlines]{lineno} 

\usepackage{tikz}
\usetikzlibrary{matrix}
\usetikzlibrary{backgrounds}
\usetikzlibrary{calc}
\usetikzlibrary{arrows,shapes}
\usetikzlibrary{decorations.pathmorphing}

\newtheorem{theorem}{Theorem}
\newtheorem{lemma}{Lemma}

\newtheorem{assumption}{Assumption}

\newcommand{\argmax}{\operatornamewithlimits{argmax}}

\newcommand{\mb}{\mathbf}
\newcommand{\mbb}{\mathbb}
\newcommand{\mc}{\mathcal}
\newcommand{\bs}{\boldsymbol}
\newcommand{\E}{\mbb{E}}

\newcommand{\argmin}{\operatornamewithlimits{argmin}}

\title{Ancestral state reconstruction with large numbers of sequences and edge-length estimation}
\date{\today}
\author{
Lam Si Tung Ho \\
Department of Mathematics and Statistics \\
Dalhousie University, Halifax, Nova Scotia, Canada
\and
Edward Susko \\
Department of Mathematics and Statistics \\
Dalhousie University, Halifax, Nova Scotia, Canada
}

\begin{document}
\maketitle

\begin{abstract}
Likelihood-based methods are widely considered the best approaches for reconstructing ancestral states. Although much effort has been made to study properties of these methods, previous works often assume that both the tree topology and edge lengths are known. In some scenarios the tree topology might be reasonably well known for the taxa under study. When sequence length is much smaller than the number of species, however, edge lengths are not likely to be accurately estimated. We study the consistency of the maximum likelihood and empirical Bayes estimators of ancestral state of discrete traits in such settings under a star tree. We prove that the likelihood-based reconstruction is consistent under symmetric models but can be inconsistent under non-symmetric models. We show, however, that a simple consistent estimator for the ancestral states is available under non-symmetric models. The results illustrate that likelihood methods can unexpectedly have undesirable properties as the number of sequences considered get very large. Broader implications of the results are discussed.

\end{abstract}

\section{Introduction}

Ancestral state reconstruction is an important problem in evolutionary biology
\citep{maddison1994ancestral,felsenstein2004book,liberles2007book}.
Reconstructing the ancestral states helps answer many questions about macroevolution including the evolution of phenotypes \citep{finarelli2006ancestral, odom2014female} and the origin of epidemics \citep{lemey2009bayesian, faria2014early, gill2017relaxed}. Ancestral reconstruction has also been used to determine which types of substitutions frequently occur in pseudogenes \citep{gojobori1982pseudogenes} and to study the optimal growth
temperature of the Last Universal Common Ancestor \citep{gaucher2008temperature}. 
More broadly, the degree to which ancestral sequences can be accurately reconstructed is indirectly related to the efficiency of tree reconstruction methods, molecular dating and inference about adaptive evolution in molecular settings. This is because such methods, in effect, consider weighted averages of probabilities over ancestral sequences. 

With a relatively large number of sites and a small to moderate number of sequences, edge lengths and even tree topology can be estimated accurately from sequence data. For this reason and for simplicity, previous studies on the theory of ancestral state reconstruction methods often assume that the tree topology and edge lengths are known \citep{ane2008analysis, royer2013choosing, fan2018necessary, ho2019multi}. However, large numbers of sequences are increasingly available for a wide variety of species,
and phylogenies based on hundreds or thousands of taxa are becoming commonplace. Increasing the number of taxa gives more information but to avoid having large amounts of missing data, the number of characters considered often needs to be kept small. Some theoretical results are available in such settings and suggest challenges, particularly for edge-length estimation. For instance, for $n$ species, the required sequence length for accurately reconstruct the tree topology is a power of $\log n$ \citep{erdHos1999few, erdos1999fewII} while the required sequence length for reconstructing both the tree topology and edge lengths is a power of $n$ \citep{dinh2018consistency}. In such settings, treating edge-lengths as known is problematic yet little effort has been made to study the problem of ancestral state reconstruction without edge lengths, especially for likelihood-based methods. In this paper, we will focus on this problem for discrete traits. 

The simplest ancestral reconstruction method for discrete traits is Majority rule, which estimates the state at the root by the most frequent state appearing at the leaves.
Maximum parsimony, on the other hand, utilizes the information from the tree topology.
This method estimates the root value by minimizing the number of changes needed to explain the evolution of the character along the tree. Maximum parsimony can have strong biases in the presence of compositional bias, however 
\citep{collins1994compositional,eyre-walker1998parsimony}.
The maximum likelihood estimator (MLE) and Bayesian inference maximize the likelihood function and the posterior distribution respectively for reconstructing the ancestral state.
These likelihood-based methods employ the information from both the tree topology and edge lengths. Intuitively, utilizing more information can be expected to result in a more efficient estimation method. Moreover, in standard settings, likelihood methods are known to have the opimality property of being asymptotically minimum variance among approximately unbiased estimators \cite[\S 5.4.3]{bickel2007book}.
Therefore, it is not surprising that they are often considered the best approaches for ancestral state reconstruction. The setting considered here however is non-standard at least in that the number of parameters increase as the number of taxa increase. 

Consistency is often considered a base criterion for judging whether an estimation method is good or not.
An estimator is consistent if it converges to the true value as the number of observations increases to infinity. In the present setting, consistency arises if 
we can recover the true ancestral state when we have an infinite number of species.
When the tree topology and edge lengths are known, \citet{fan2018necessary} provides a necessary and sufficient condition, called ``big bang'', for the existence of a consistent estimator of the ancestral state on bounded-height trees.
It is worth noticing that a direct consequence of Proposition 6 in \citet{steel2008maximum} is that the MLE is consistent if there exists a consistent estimator for the ancestral state (assuming that the evolution model is known).
This result confirms that the MLE is a reasonable estimator of the ancestral state in this scenario and that the ``big bang'' condition is a necessary and sufficient condition for the consistency of the MLE.

A natural hypothesis is that likelihood-based methods are also the best ancestral state reconstruction methods when edge lengths are unknown.
To investigate this, we consider a simple scenario where discrete traits evolve along a star tree according to a proportional model. Although estimation of ancestral frequencies is usually of greatest interest under non-stationary models
\citep{susko2013ancestral}, we assume a simple stationary setting. We show that in this setting the MLE and Empirical Bayes estimator (sometimes also referred to as the Maximum A Posteriori (MAP) estimator) converge upon the same solution. Consequently, it suffices to consider the MLE. 
The MLE is shown to be consistent under symmetric models but there exists a zone of inconsistency under non-symmetric models.
As a consequence, the ``big bang'' condition in \citet{fan2018necessary} is no longer a sufficient condition for the consistency of the MLE.
We also uncover that when the edge lengths are unknown, the MLE is not the best ancestral reconstruction method.
Specifically, we present a simple new estimator for the ancestral state that is consistent under some mild conditions.
We show that the MLE is not always consistent under the same conditions.
Therefore, the proposed estimator is better than the MLE in this scenario.

\section{Settings}

Throughout this paper, we will focus on star trees whose edge lengths are unknown.
A star tree is a tree such that all taxa are direct descendants of the root (Figure \ref{fig:star-tree}).
Let $n$ be the number of leaves and $\mb t = (t_k)_{k=1}^n$ be the (unknown) edge lengths.
In this setting, we observe a sequence of $N$ sites at each leave.
We assume that characters at these sites evolve independently along the tree according to the \emph{proportional model}.
That is, the evolution of characters follows a finite-state continuous-time Markov process with the following transition probabilities
\[
P_{ij}(t) = \pi_j [1 - \exp (-\mu t)] + 1_{\{i = j\}}\exp (-\mu t), \quad \forall i,j \in \{1,2, \ldots,  c\}.
\]
Here, $c$ is the number of possible states and $\bs \pi = (\pi_1, \pi_2, \ldots, \pi_c)$ is the stationary distribution of the process. In practice, to avoid problems of confounding, $\mu=[\sum_j \pi_j (1-\pi_j)]^{-1}$ is used so that edge lengths are interpretable as expected numbers of substitutions. We consider the re-parameterization $\mb s = \exp[-\mu \mb t]$.
So, the transition probabilities become
\[
P_{ij}(s) = \pi_j [1 - s] + 1_{\{i = j\}}s.
\]

\begin{figure}[t]
  \centering
  \includegraphics[width=0.5\textwidth]{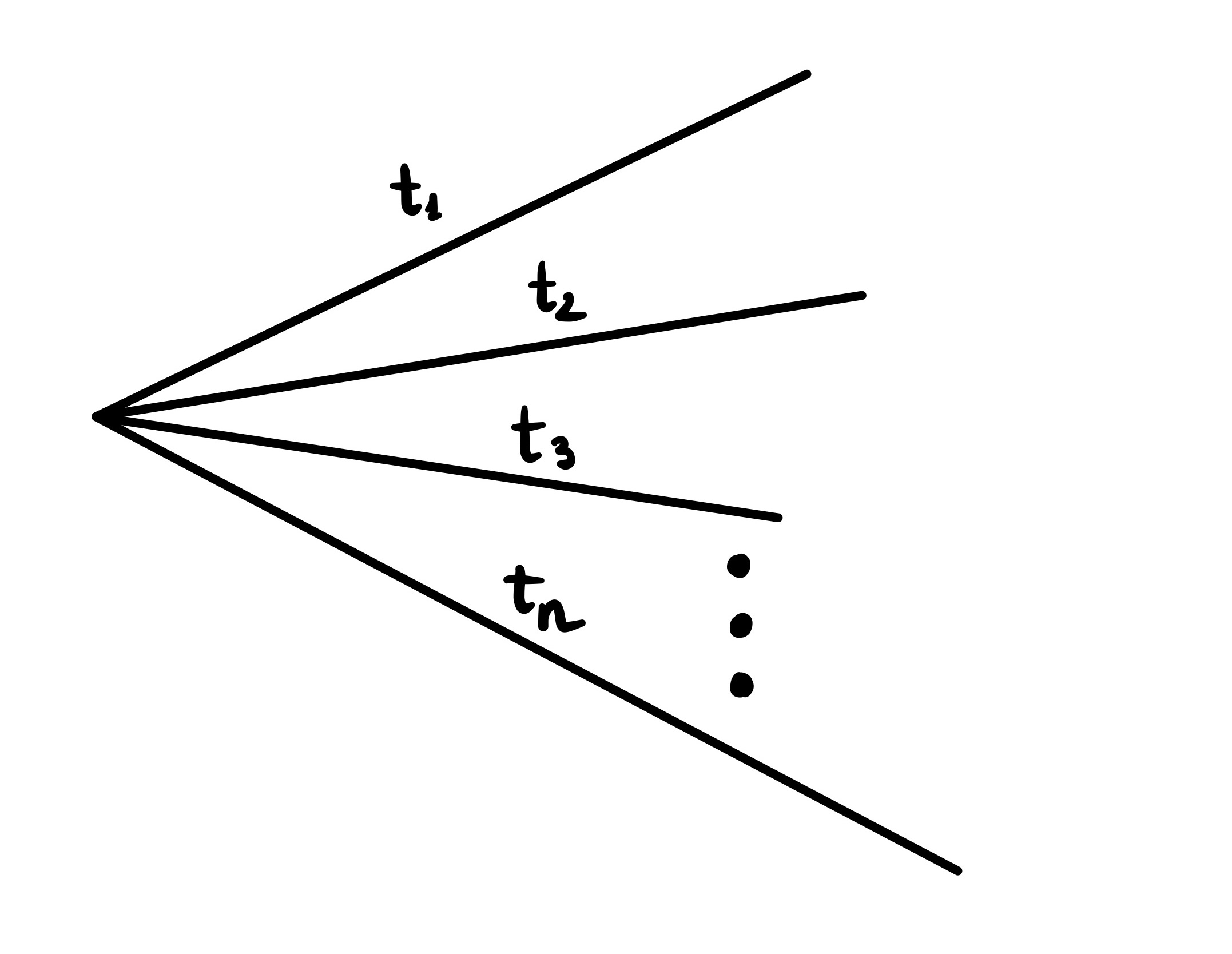}
\caption{A star tree.}
\label{fig:star-tree}
\end{figure}

We are interested in reconstructing the ancestral states $\bs \rho^*$ at the root of the tree.
Without loss of generality, we can assume that
\[
\bs \rho^*=(\underbrace{1,\dots,1}_{N_1},\dots, \underbrace{c,\dots,c}_{N_c}).
\]
When the edge lengths are unknown, $(\bs \rho^*, \bs \pi)$ are not identifiable under the proportional model \citep{gascuel2020darwinian}.
Therefore, unless mentioned otherwise, we assume that the stationary distribution $\bs \pi>\bs 0$ is known. The setting is a proxy for the situation where frequencies can be estimated from a larger tree and the star portion is a local subtree of focus for ancestral reconstruction. 
We also make the following assumption.
\begin{assumption}
Define
\[
\bar{s} = \frac{1}{n}\sum_{k=1}^n s_k.
\]
We assume that $\lim\inf \bar{s} > 0$.
\label{assump:lims}
\end{assumption}

\noindent 
The reason for Assumption \ref{assump:lims} is to guarantee that the majority of edge lengths are not so large ($s=0$ corresponds to $t=\infty$) that the tip data is almost independent of the root data. A trivial scenario where Assumption \ref{assump:lims} is satisfied is when $t_k = t $ for all $k \in \mbb{N}$. We describe a non-trivial example where these assumptions hold in the following Lemma.

\begin{lemma}
Let $t_1, t_2, \ldots, t_n$ be independent and identically (iid) distributed random variables on $\mbb{R}^+$ with finite mean.
Then, the star tree with edge lengths $\mb t = (t_k)_{k=1}^n$ satisfies Assumption \ref{assump:lims}.
\label{lem:ex}
\end{lemma}

For two sequences $\bs \rho = (\rho_1, \rho_2, \ldots, \rho_N)$ and $\mb y = (y_1, y_2, \ldots, y_N)$, we define
\[
P_{\bs \rho \mb y}(s) = \prod_{l = 1}^N{P_{\rho_l y_l}(s)}, \quad \text{and} \quad \bar{P}_{\bs \rho \mb y} = \frac{1}{n}\sum_{k=1}^n{P_{\bs \rho \mb y}(s_k)}.
\]
Note that $\bar{s^k}=\int s^k ~ dF_{n}(s)$ where $F_n$ is the distribution function placing weight $1/n$ on $s_1,\dots, s_n$. A sequence of distribution functions always has a convergent subsequence. We consider a convergent subsequence in what follows. For such a subsequence, because $s\in [0,1]$, $\lim_n \bar{s^k}$ exists. Since $\bar P_{\bs \rho^* \mb y}$ is a linear transformation of the $\bar{s^k}$, with fixed coefficients that do not change with $n$, its limit exists for this convergent subsequence. The linear transformation $\bar P_{\bs \rho^* \mb y}$ has positive coefficients. Thus this limit, which we denote $p_{\bs \rho^* \mb y},$ is also positive.

The restriction to a convergent subsequence is without loss of generality. In cases where inconsistency is shown for a convergent subsequence it is implied for the original sequence. In cases where consistency is shown, because it holds for an arbitrary subsequence, by the Helly subsequence principle it holds for the original sequence. 
Therefore, without loss of generality, we can assume that $\bar s$ converges to a positive number and $\bar P_{\bs \rho^* \mb y}$ converges to $p_{\bs \rho^* \mb y}$.




\section{Ancestral state reconstruction}

The majority rule and maximum parsimony reconstruction methods have been studied extensively in the scenario where edge lengths are unknown because these methods do not take into account this information \citep{maddison1995calculating, gascuel2010inferring, mossel2014majority, herbst2018accuracy, herbst2019quantifying}.
On the other hand, little is known about asymptotic properties of likelihood-based methods in this scenario.

\subsection{Maximum likelihood estimator}

Given its good performance in a wide variety of settings, it a reasonable hypothesis that the MLE is a good method for reconstructing the ancestral state.
Since the edge lengths are unknown, we need to estimate both the ancestral states and edge lengths jointly.
Let $\mb y_k$ be the sequence observed at the $k$-th leave.
The MLE is defined by
\[
\hat{\bs \rho} = \argmax_{\bs \rho} \left ( \max_{\mb s} \sum_{k=1}^n{\log P_{\bs \rho \mb y_k}(s_k)} \right ).
\]
For an $N$-dimensional pattern $\mb y$, let $n_{\mb y}$ be the number of times $\mb y$ is observed at the leaves.
That is, 
\[
n_{\mb y} = \sum_{k=1}^n{1_{\{\mb y_k = \mb y \}}}.
\]
We define
\[
\ell(\bs \rho)  = \sum_{\mb y} n_{\mb y} \log P_{\bs \rho \mb y}(\hat s(\bs \rho, \mb y))
\]
where the sum is over all possible patterns, and $\hat s(\bs \rho, \mb y) = \argmax_s P_{\bs \rho \mb y}(s)$. 
Then, we have
\[
\hat{\bs \rho} = \argmax_{\bs \rho} \frac{1}{n}\ell(\bs \rho).
\]
Then
\begin{equation}
\E \left [\frac{1}{n}\ell(\bs \rho) \right ] = \sum_{\mb y} {\bar{P}_{\bs \rho^* \mb y} \log P_{\bs \rho \mb y} (\hat s(\bs \rho, \mb y))} \to \sum_{\mb y} {p_{\bs \rho^* \mb y} \log P_{\bs \rho \mb y} (\hat s(\bs \rho, \mb y))}.
\label{eqn:func}
\end{equation}
We define
\[
e(\bs \rho) = \sum_{\mb y} {p_{\bs \rho^* \mb y} \log P_{\bs \rho \mb y} (\hat s(\bs \rho, \mb y))}.
\]

The limit \eqref{eqn:func} prompts an immediate question: Is there connection between the MLE $\hat{\bs \rho}$ and the function $e(\bs \rho)$?
The answer is yes.
When the number of leaves is large, the MLE is a maximizer of $e(\bs \rho)$ with high probability.
Specifically,

\begin{lemma}
Let $\mc{H}$ be the set of maximum points of $e(\bs \rho)$. Then, under Assumption \ref{assump:lims},  
\[
\lim_{n \to \infty}\Pr(\hat{\bs \rho} \in \mc{H}) = 1.
\] 
\label{lem:MLE}
\end{lemma}

\noindent 
A direct consequence of Lemma \ref{lem:MLE} is that if $\bs \rho^* \notin \mc{H}$, then the MLE is inconsistent.
Another one is that if $e(\bs \rho)$ have a unique maximizer $\bs \rho_M$, then the MLE of the ancestral sequence converges to $\bs \rho_M$.



\subsection{Empirical Bayes estimator}

The empirical Bayes estimator is more widely used in practice than the MLE. For estimation of edge-lengths, it uses the more conventional likelihood that averages over the unobserved sequences at internal nodes:
\[
\hat{\mb s} = \argmax_{\mb s} \sum_{\bs \rho} P_{\bs \pi}(\bs \rho) \prod_{k=1}^n P_{\bs \rho \mb y_k}(s_k)
\]
where $P_{\bs \pi}$ is the probability according to the stationary distribution $\bs \pi$. The estimator of the ancestral sequence is then obtained as the sequence having the largest conditional probability, given the data and calculated using the estimated edge-lengths:
\[
\hat{\bs \rho}_B = \argmax_{\bs \rho} P(\bs \rho \mid (\mb y_k)_{k=1}^n, \hat{\mb s}) = \argmax_{\bs \rho} P_{\bs \pi}(\bs \rho) \prod_{k=1}^n P_{\bs \rho \mb y_k}(\hat s_k) . 
\]
\noindent 
Here, the empirical Bayes estimator estimates edge lengths first, then we estimate ancestral states conditional on these estimated values.
Will this make any difference?
A short answer is no.
Just like the MLE, the empirical Bayes estimator will eventually be in the set $\mc{H}$ of maximizers of $e(\bs \rho)$.
\begin{theorem}
Under Assumption \ref{assump:lims},
\[
\lim_{n \to \infty}\Pr(\hat{\bs \rho}_B \in \mc{H}) = 1.
\] 
\label{thm:EmB}
\end{theorem}

\noindent
As a consequence, both the empirical Bayes estimator and MLE will be consistent if $\mc{H}=\{\bs \rho^*\}$ and will be inconsistent if $\bs \rho^*\notin \mc{H}$. 

For a fixed ancestral sequence $\bs \rho$ and an observed sequence $\mb y$, let $n_{\bs \rho \mb y}(i,j)$ be the number of times the putative ancestral state was $\rho_l=i$ with corresponding observed value equal to $y_l=j$. 
We have 
\begin{align} 
\nonumber
\log P_{\bs \rho \mb y}(s)&=\sum_{i = 1}^c \sum_{j=1}^c n_{\bs \rho \mb y}(i,j) \log P_{i j}(s) \\
\nonumber
&=\sum_{i=1}^c n_{\bs \rho \mb y}(i,i) \log[\pi_{i}+ (1-\pi_{i})s] + \sum_{i=1}^c\sum_{j \mid j\neq i} n_{\bs \rho \mb y}(i,j) \log[\pi_j (1-s)]\\
\nonumber
&= \sum_{i=1}^c n_{\bs \rho \mb y}(i,i) \log[\pi_{i}+ (1-\pi_{i})s] + \left[ N-\sum_{i=1}^c n_{\bs \rho \mb y}(i,i) \right ] \log(1-s)  \\
\nonumber
& + \sum_{i=1}^c\sum_{j=1}^c n_{\bs \rho \mb y}(i,j) \log(\pi_j) - \sum_{i=1}^c n_{\bs \rho \mb y}(i,i) \log(\pi_{i}).
\end{align}

The third term depends on $\bs \rho$ and $\mb y$ through $\sum_i  n_{\bs \rho \mb y}(i,j)$, which is the number of $y_l=j$. Since it depends on $\mb y$ alone, letting $C(\mb y) = \sum_{i=1}^c\sum_{j=1}^c n_{\bs \rho \mb y}(i,j) \log(\pi_j)$, we have that
\begin{align}
\nonumber
\log P_{\bs \rho \mb y}(s) &= \sum_{i=1}^c n_{\bs \rho \mb y}(i,i) \log[\pi_{i}+ (1-\pi_{i})s] + \left[ N-\sum_{i=1}^c n_{\bs \rho \mb y}(i,i) \right ] \log(1-s)  \\
\label{eq:prop1}
& - \sum_{i=1}^c n_{\bs \rho \mb y}(i,i) \log(\pi_{i}) + C(\mb y).
\end{align}

Maximizing $\log P_{\bs \rho \mb y}(s)$, we obtain:

\begin{lemma}
\begin{itemize}
\item If 
\[\sum_{i=1}^c \frac{n_{\bs \rho \mb y}(i,i)}{\pi_{i}} \le N,
\]
then $\hat s(\bs \rho, \mb y) = 0$.
\item If 
\[
\sum_{i=1}^c n_{\bs \rho \mb y}(i,i) = N,
\]
then $\hat s(\bs \rho, \mb y) = 1$.
\item Otherwise,
\[
\sum_{i=1}^c  \frac{n_{\bs \rho \mb y}(i,i)}{\pi_{i}+ (1-\pi_{i}) \hat s(\bs \rho, \mb y)} = N.
\]
\end{itemize}
\label{lem:MLEedge}
\end{lemma}

First, let us consider symmetric models, that is $\pi_1 = \pi_2 = \ldots = \pi_c = 1/c$.
In this scenario, the MLE for the ancestral states is consistent.
It is worth noticing that under symmetric models, the MLE and Maximum parsimony are the same when there is only $1$ site ($N = 1$) \citep{tuffley1997links}.

\begin{theorem}
Suppose that Assumption \ref{assump:lims} holds. For symmetric models, we have
\[
\lim_{n \to \infty}\Pr(\hat{\bs \rho} = \bs \rho^*) = 1.
\]
\label{thm:symmetric}
\end{theorem}

Theorem \ref{thm:symmetric} suggests that we may be able to reconstruct the ancestral states with high accuracy even when the edge lengths are unknown.
This result aligns with similar findings in \citet{gascuel2010inferring}.
Since the MLE is consistent under symmetric models, would it also work well under non-symmetric models?
Unfortunately, it is not the case.
To see this, let us consider a single site scenario (i.e. $N = 1$).
Suppose $\rho^* = r$.
For any $\rho = a$, we have
\[
\hat s(a, y) =
\begin{cases}
1 & \text{if } y = a \\
0 & \text{if } y \ne a,
\end{cases}
\quad
\text{and}
\quad
\log P_{\rho y}(\hat s(a, y)) =
\begin{cases}
0 & \text{if } y = a \\
\log(\pi_y) & \text{if } y \ne a,
\end{cases}
\]
By \eqref{eqn:func}, we have
\[
e(a) = \sum_{y \ne a} p_{ry} \log(\pi_y) = \left (\sum_{y} p_{ry} \log(\pi_y)  \right )- p_{ra} \log(\pi_a)
\]
By Lemma \ref{lem:MLE}, the MLE is inconsistent if
\begin{align}
\nonumber
e(r) < e(a) &\Longleftrightarrow  p_{ra} \log(\pi_a) < p_{rr} \log(\pi_r) \\
\nonumber
&\Longleftrightarrow \lim_n [1 - \bar s] \pi_a \log(\pi_a) < \lim_n [\pi_r + (1 - \pi_r) \bar s]\log(\pi_r) \\
\label{eqn:inconsistency_cond}
&\Longleftrightarrow \lim_n \bar s< \frac{\pi_r\log(\pi_r)-\pi_a\log(\pi_a)}{\pi_r\log(\pi_r)-\pi_a\log(\pi_a)-\log(\pi_r)}.
\end{align}

\noindent
It is worth noticing that \eqref{eqn:inconsistency_cond} cannot hold under symmetric models ($\pi_r = \pi_a$), which aligns with Theorem \ref{thm:symmetric}.
On the other hand, when $\pi_r\log(\pi_r) > \pi_a\log(\pi_a)$, \eqref{eqn:inconsistency_cond} holds with $\lim_n \bar s$ is sufficiently small.
Since the function $\pi\log(\pi)$ is maximized at $\pi=e^{-1}$, there exists a zone of inconsistency of the MLE unless $\pi_a=e^{-1}$.
We generalize this argument to obtain the zone of inconsistency for the case when we have more than $1$ site.

\begin{theorem}
Assume that $\rho^* = (r, \dots, r)$. Then, there exists $\lim_n \bar s > 0$ and $\bs \pi$ such that the MLE is inconsistent.
A sufficient condition for inconsistency is that $v(\pi_a) > v(\pi_r)$ where
\[
v(p)=E\left[\left\{~\hat p \log\bigg[\frac{\hat p}{p} \bigg]+
   (1-\hat p) \log\bigg[\frac{1-\hat p}{1-p}\bigg]~\right\}
   ~I\big\{\hat p > p \big\}\right] ,
   \]
calculated with $\hat p=X/N$ and $X\sim\mbox{binomial}(N,p).$ 
\label{thm:zone}
\end{theorem}
\noindent
Figure \ref{fig:hm-ancestral} visualizes the zone of inconsistency described in Theorem \ref{thm:zone}. Perhaps surprisingly, with a small number of sites, a preferred alternative ancestral state can have a smaller stationary frequency than the true ancestral state. As the sequence length gets larger, however, The region approaches the region for Parsimony where the alternative stationary frequency is larger than that of the true ancestral state. 

\begin{figure}[t]
  \centering
  \includegraphics[width=0.7\textwidth]{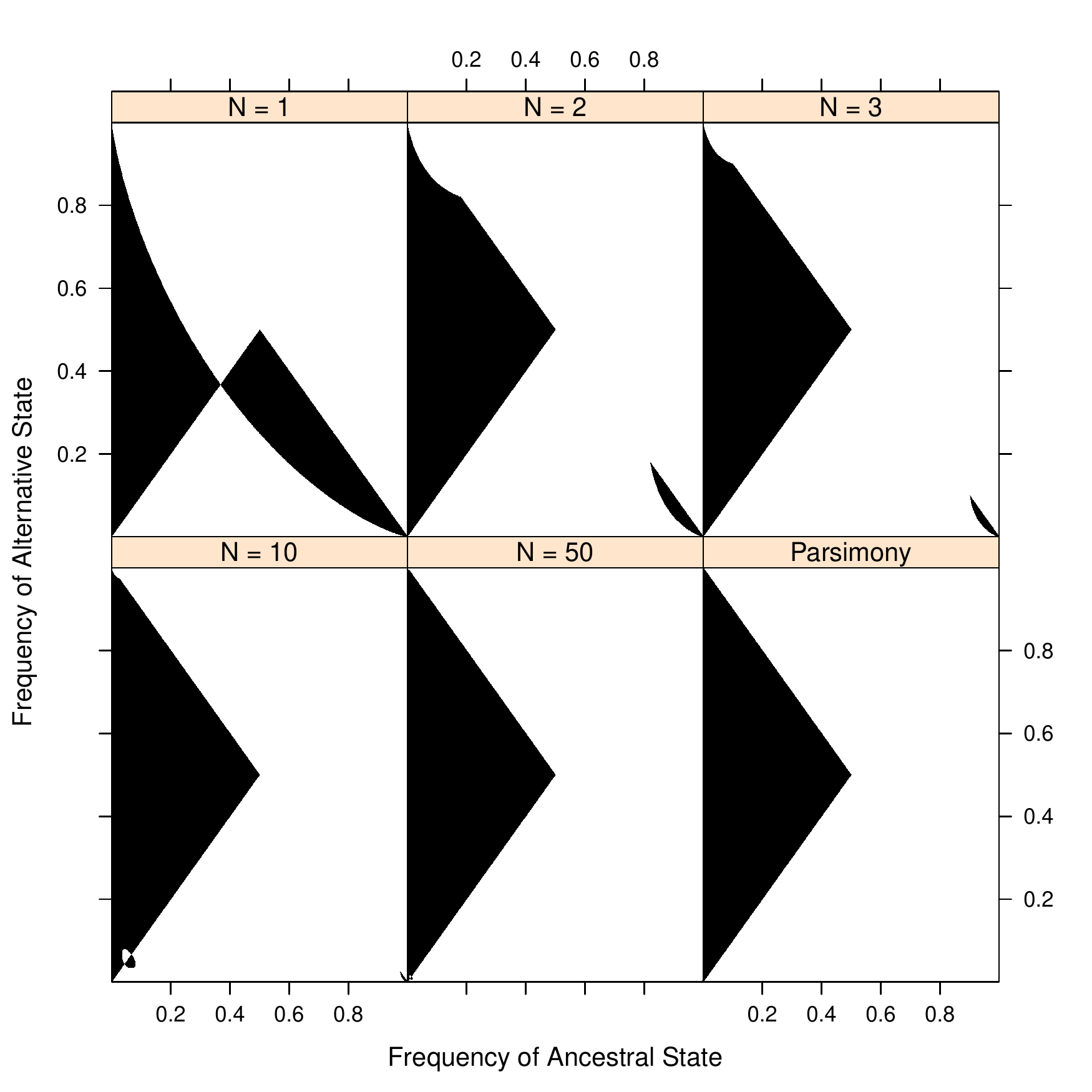}
\caption{Areas in black indicate frequencies of pairs of true ancestral character states and alternative character states for which there exists $\lim \bar s>0$ such that the alternative state is preferred to the true ancestral state.}
\label{fig:hm-ancestral}
\end{figure}

\subsection{A simple consistent estimator}

Although Theorem \ref{thm:zone} seems to eliminate our hope of reconstructing the ancestral states accurately under non-symmetric models, we note that this negative result was only shown to apply to the case when the ancestral state is the same for all sites.
However, such a case is rare in practice when the number of sites is large.
Assuming that there are  at least two distinct ancestral states, we propose a simple estimator that can estimate the ancestral states consistently.
Let $y_{kl}$ be the character state at site $l$ of the $k$-th species.
For $j \in \{ 1, 2, \ldots, c \}$, denote
\[
\hat \pi_{jl} = \frac{1}{n} \sum_{k=1}^n{I\{y_{kl}=j\}}, \quad V(j,l)=\hat\pi_{jl} - \sum_{v\neq l} \frac{\hat\pi_{jv}}{N-1}.
\]
We estimate the ancestral state at site $l$ by 
\[
\rho_l^{(D)} = \arg\max_{j} V(j,l).
\]
The motivation for this estimate is that for the proportional model, 
\begin{equation}
\label{eq:dchar-1}
E[\hat\pi_{jl}] = \frac{1}{n}\sum_{k=1}^n \left ( \pi_j - \pi_j s_k + I\{\rho_l^* = j\} s_k \right ) = \pi_j - \pi_j \bar s + I \{\rho_l^* = j\} \bar s.
\end{equation}
It follows that the observed frequency of $\rho_l^*$ at site $l$ is expected to be larger than the average observed frequency of $\rho_l^*$ over other sites, unless they all have $\rho_l^*$ as their ancestral character state. 

\begin{theorem}
If Assumption \ref{assump:lims} holds, $N = O(\sqrt{n/\log n})$, and there are at least two distinct ancestral states (that is, $\rho_l^*\neq \rho_1^*$ for some site $l$). 
Then, 
\[
\Pr(\bs \rho_l^{(D)} = \bs \rho^*) \to 1.
\]
\label{thm:distinct}
\end{theorem}
\noindent
A natural question is whether this Theorem also holds for the MLE.
Interestingly, this is not the case.
In other words, there exists a zone of inconsistency for the MLE.
Hence, the MLE is not the best ancestral reconstruction method when edge lengths are unknown and the number of sites is small relative to the number of sequences.

To see this, let us consider situation when the true ancestral states $\bs \rho^*$ has the form $(\rho^*_1, \rho^*_2 )$ such that $\rho^*_1 \ne \rho^*_2$.
Again, we want to investigate the set of maximum points of $e(\bs \rho)$.
For this simple scenario, we can derive an analytic formula for $E[\ell(\rho_1,\rho_2)/n]$.

\begin{lemma}
Denote 
\[
f(\pi)=I\{\pi< 1/2\}\log[4\pi(1-\pi)].
\]
We have
\[
E\left[\frac{\ell(\rho_1,\rho_2)}{n}\right] = C_0 - \bar P_{\rho_1^*\rho_1} f(\pi_{\rho_1}) - \bar P_{\rho_2^*\rho_2 } f(\pi_{\rho_2}) + \bar P_{(\rho_1^*, \rho_2^*) (\rho_1, \rho_2)} \{f(\pi_{\rho_1})+ f(\pi_{\rho_2})
                         -\log[\pi_{\rho_1}\pi_{\rho_2}]\}
\]
where 
\[
C_0 = \sum_{x} [\bar P_{\rho_1^* x } + \bar P_{\rho_2^* x}] \log[\pi_x]
\]  is independent of $(\rho_1,\rho_2)$. 
\label{lem:twosites}
\end{lemma}

For the proportional model, $P_{xy}(s) =(\delta_{xy}-\pi_y) s + \pi_y$, so
\[
\bar P_{(\rho_1^*, \rho_2^*) (\rho_1, \rho_2)} = (\delta_{\rho_1^*\rho_1}-\pi_{\rho_1})(\delta_{\rho_2^*\rho_2}-\pi_{\rho_2}) \overline{s^2} + [(\delta_{\rho_1^*\rho_1}-\pi_{\rho_1})\pi_{\rho_2} + (\delta_{\rho_2^*\rho_2}-\pi_{\rho_2}) \pi_{\rho_1}] \bar s + \pi_{\rho_1} \pi_{\rho_2}
\]
where
\[
\overline{s^2} = \frac{1}{n}\sum_{k=1}^n s_k^2.
\]
By Lemma \ref{lem:twosites}, we have
\begin{equation}
E\left [\frac{\ell(\rho_1,\rho_2)}{n}\right] = C_0 + A(\rho_1,\rho_2) \overline{s^2} + B(\rho_1,\rho_2) \bar s + C(\rho_1,\rho_2)
\label{eq:ts-10}
\end{equation}
where
\begin{equation}
  \label{eq:ts-11}
  A(\rho_1,\rho_2) = (\delta_{\rho_1^*\rho_1}-\pi_{\rho_1})(\delta_{\rho_2^*\rho_2}-\pi_{\rho_2})
   \{f(\pi_{\rho_1})+ f(\pi_{\rho_2})-\log[\pi_{\rho_1}\pi_{\rho_2}]\}
\end{equation}
\begin{eqnarray}
\nonumber
  B(\rho_1,\rho_2) &=& -(\delta_{\rho_1^*\rho_1}-\pi_{\rho_1})f(\pi_{\rho_1})
-(\delta_{\rho_2^*\rho_2}-\pi_{\rho_2})f(\pi_{\rho_2})\\
\nonumber
&&~~~+ (\delta_{\rho_1^*\rho_1}-\pi_{\rho_1}) \pi_{\rho_2}
    \{f(\pi_{\rho_1})+ f(\pi_{\rho_2})-\log[\pi_{\rho_1}\pi_{\rho_2}]\}\\
  \label{eq:ts-12}
&&~~~+ (\delta_{\rho_2^*\rho_2}-\pi_{\rho_2})\pi_{\rho_1}
    \{f(\pi_{\rho_1})+ f(\pi_{\rho_2})-\log[\pi_{\rho_1}\pi_{\rho_2}]\}
\end{eqnarray}
and
\begin{equation}
  \label{eq:ts-13}
  C(\rho_1,\rho_2) = - \pi_{\rho_1}f(\pi_{\rho_1}) 
    - \pi_{\rho_2}f(\pi_{\rho_2}) 
+ \pi_{\rho_1}\pi_{\rho_2}\{f(\pi_{\rho_1})+ f(\pi_{\rho_2})-\log[\pi_{\rho_1}\pi_{\rho_2}]\}
\end{equation}
Thus the difference in expected log likelihoods $Q_{\rho_1\rho_2}(\mb s):=E[\ell(\rho_1,\rho_2)/n] -E[\ell(\rho_1^*,\rho_2^*)/n]$ satisfies that 
\[
Q_{\rho_1\rho_2}(\mb s) = A_D(\rho_1,\rho_2) \bar s^2 + B_D(\rho_1,\rho_2) \bar s + C_D(\rho_1,\rho_2),
\]
where $A_D(\rho_1,\rho_2) = A(\rho_1,\rho_2)-A(\rho_1^*,\rho_2^*), B_D(\rho_1,\rho_2) = B(\rho_1,\rho_2) - B_D(\rho_1^*,\rho_2^*), C_D(\rho_1,\rho_2) = C(\rho_1,\rho_2) - C(\rho_1^*,\rho_2^*)$.  
If $Q_{\rho_1\rho_2}(\mb 0) = C_D(\rho_1,\rho_2)>0$, then for $\overline{s^2}$ sufficiently small, the expected log likelihood for $(\rho_1,\rho_2)$ is larger than for $(\rho_1^*,\rho_2^*)$ and estimation is inconsistent. 
Note that $(\bar s)^2 \leq \overline{s^2}$.

For simplicity, let us focus on the case that all edge-lengths are equal.
That is, $s_k = s$ for $k = 1, 2 ,\ldots, n$.
In this situation,
\[
Q_{\rho_1\rho_2}(s) = A_D(\rho_1,\rho_2) s^2 + B_D(\rho_1,\rho_2) s + C_D(\rho_1,\rho_2).
\]
Note that $P_{\rho_j^*\rho_j}(1)=\delta_{\rho_j^*\rho_j}$.
So, Lemma \ref{lem:twosites} gives that $Q_{\rho_1\rho_2}(1)=\log[\pi_{\rho_1^*}\pi_{\rho_2^*}] < 0$ with $\rho_1^*\neq\rho_1$ and $\rho_2^*\neq\rho_2$.
If $\rho_1^*\neq\rho_1$ but $\rho_2^*=\rho_2$ then
$Q_{\rho_1\rho_2}(1)= - f(\pi_{\rho_2^*})+\log[\pi_{\rho_1^*}\pi_{\rho_2^*}].$ 
If $\pi_{\rho_2^*}\ge 1/2$, $f(\pi_{\rho_2^*})=0$ and $Q_{\rho_1\rho_2}(1)=\log[\pi_{\rho_1^*}\pi_{\rho_2^*}]<0$. If $\pi_{\rho_2^*}<1/2$, then 
$$Q_{\rho_1\rho_2}(1)=-\log[4\pi_{\rho_2^*}(1-\pi_{\rho_2^*})] + \log[\pi_{\rho_1^*}\pi_{\rho_2^*}]  < -\log[4(1-\pi_{\rho_2^*})] < 0.$$
Thus $Q_{\rho_1\rho_2}(1)<0$ for $\rho_1^*\neq\rho_1$ but $\rho_2^*=\rho_2$. By symmetry $Q_{\rho_1\rho_2}(1)<0$ for $\rho_1^*=\rho_1$ but $\rho_2^*\neq\rho_2$. In summary, $Q_{\rho_1\rho_2}(1)<0$ for $(\rho_1,\rho_2)\neq(\rho_1^*,\rho_2^*).$

If $Q_{\rho_1\rho_2}(0)>0$, then since $Q_{\rho_1\rho_2}(1)<0$ and $Q_{\rho_1\rho_2}(s)$ is quadratic, $Q_{\rho_1\rho_2}(s)$ will have exactly one root in $s^* \in [0,1].$ 
Moreover, for $s < s^*$, $Q_{\rho_1\rho_2}(s) > 0$ and $Q_{\rho_1\rho_2}(s) < 0$ for $s > s^*$. 
Thus, the MLE is inconsistent for $s < s^*$ whenever there exists $(\rho_1,\rho_2)$ with $Q_{\rho_1\rho_2}(0)>0$. 
In Table \ref{tab:ts-1} we tabulate the values of $\pi$ in the intersection of the unit simplex and $\{0.1,\dots,0.9\}^4$ for which the MLE can be inconsistent. 
The $(\rho_1,\rho_2)$ indicated are the ones giving the largest $Q_{\rho_1\rho_2}(0)$. 
In the case that all edge-lengths are equal, $s = \exp[-\mu t],$ so estimation is inconsistent $t > t^*:=-\log(s^*)/\mu.$

The above discussion assumed constant $s_k=s$. In the cases that the $s_k$ vary,
\[
Q_{\rho_1\rho_2}(\mb s) = A_D(\rho_1,\rho_2)[ \bar{s^2} -(\bar s^2)] + A_D(\rho_1,\rho_2)(\bar s)^2 + B_D(\rho_1,\rho_2) \bar s + C_D(\rho_1,\rho_2),
\]
If $\lim \bar s$ corresponds to a value of $s$ that gives rise to inconsistency in the single $s_k$ case, then for sufficiently small variation of the $s_k$ (small $\bar{s^2}-\bar s^2$) inconsistency will arise. However, in the calculations leading to Table \ref{tab:ts-1} we always found that $A_D(\rho_1,\rho_2)<0$. Since inconsistency only arises when $Q_{\rho_1\rho_2}(\mb s) >0$, the result suggests that trees with more variable edge-lengths are less likely to lead to inconsistency, all other things being equal. 

\begin{table}
  \centering
\caption{\label{tab:ts-1}The values of $\pi$ in the intersection of the unit simplex and $\{0.1,\dots,0.9\}^4$ for which the MLE can be inconsistent. Here $(\rho_1^*,\rho_2^*)=(A,C)$ and $(\rho_1,\rho_2)$ give ancestral character states having a larger likelihood than $(\rho_1^*,\rho_2^*)$ when all edge lengths are equal and the common edge length is larger than $t^*$. More generally, estimation is inconsistent for $ s < \exp[-t^*/\sum_i\pi_i(1-\pi_i)]$.}
  \begin{tabular}{llcllcccc}
$\pi_{A}$ & $\pi_{C}$ & &
$\pi_{G}$ & $\pi_{T}$ & &
$\rho_1$ & $\rho_2$ & $t^*$\\\hline
0.1 & 0.1 &~& 0.2 & 0.6 &~& G & T & 2.2 \\ 
0.1 & 0.1 &~& 0.3 & 0.5 &~& G & T & 2.3 \\ 
0.1 & 0.1 &~& 0.4 & 0.4 &~& G & T & 2.4 \\ 
0.1 & 0.2 &~& 0.1 & 0.6 &~& C & T & 2.3 \\ 
0.1 & 0.2 &~& 0.2 & 0.5 &~& C & T & 2.8 \\ 
0.1 & 0.2 &~& 0.3 & 0.4 &~& G & T & 2.6 \\ 
0.1 & 0.3 &~& 0.1 & 0.5 &~& C & T & 2.1 \\ 
0.1 & 0.3 &~& 0.2 & 0.4 &~& C & T & 2.4 \\ 
0.1 & 0.3 &~& 0.3 & 0.3 &~& C & G & 2.6 \\ 
0.2 & 0.2 &~& 0.1 & 0.5 &~& A & T & 3.4 \\ 
0.2 & 0.2 &~& 0.3 & 0.3 &~& G & T & 3.5 \\ 
0.1 & 0.4 &~& 0.1 & 0.4 &~& C & T & 2.1 \\ 
0.1 & 0.4 &~& 0.2 & 0.3 &~& C & T & 2.3 \\ 
0.2 & 0.3 &~& 0.1 & 0.4 &~& C & T & 3.0 \\ 
0.2 & 0.3 &~& 0.2 & 0.3 &~& C & T & 3.5 \\ 
0.1 & 0.5 &~& 0.1 & 0.3 &~& C & T & 2.0 \\ 
0.1 & 0.5 &~& 0.2 & 0.2 &~& C & G & 2.3 \\ 
0.2 & 0.4 &~& 0.1 & 0.3 &~& C & T & 2.9 \\ 
0.2 & 0.4 &~& 0.2 & 0.2 &~& A & G & 4.0 \\ 
0.1 & 0.6 &~& 0.1 & 0.2 &~& C & T & 2.0 \\ 
0.1 & 0.7 &~& 0.1 & 0.1 &~& A & G & 2.1 \\ 
  \end{tabular}
\end{table}

\clearpage
\section{Discussion and Conclusion}

Likelihood-based methods, in particular the MLE, are often considered the best methods for ancestral state reconstruction.
In this paper, we studied the consistency of the MLE for the problem of reconstructing the ancestral state of discrete traits on star trees whose edge lengths are unknown.
We proved that the MLE is consistent under symmetric models but can be inconsistent under non-symmetric models.
It is worth noticing that Theorem \ref{assump:lims} implies that the empirical Bayes estimator will have the same difficulties as the MLE. This is a little surprising given the findings of \citet{shaw2019jointML} who showed that 
when the number of sites is large relative to the number of taxa, Maximum likelihood estimation treating ancestral states as parameters can lead to inconsistent topological and edge-length estimation. By contrast the approach that averages over ancestral states, which is more analogous to empirical Bayes, does not suffer from such difficulties. We see here that with large numbers of sequences and small numbers of sites, the Maximum likelihood and empirical Bayes approaches are more comparable. 

The results were for the setting of ancestral reconstruction for the root of a star tree. This is for a simpler framework but results likely apply to multifurcations or ``big bang'' settings. Indeed, although not shown here, the results for the simple difference of frequency estimator can be extended directly to multifurcations. Outside of these settings, ancestral reconstructions will be more variable and not converge. Nevertheless, the results suggest that, more broadly, biases in ancestral reconstructions can become a significant difficulty with large numbers of sequences.  

The effects of adding taxa or sites to an existing alignment has long been of interest \citep{graybeal1998taxon,yang1998taxon,pollock2002taxon,zwickl2002taxon}. Studies tend to conclude that phylogenetic accuracy is improved by additional taxon sampling. Most such research considers relatively small numbers of taxa, however. The results here do not deal directly with phylogenetic estimation which deserves additional study. They are suggestive of potential difficulties with large numbers of taxa, however. To see this, consider comparison of two conflicting topologies differing in $S_1 S_2|S_3 S_4$ versus $S_1 S_3|S_2S_4$ where $S_1,\dots, S_4$ are subtrees each with large numbers of taxa that give rise to separate clades. The likelihood for $S_1 S_2|S_3 S_4,$ for instance, would be calculated as
\[
\sum_{x} P[x_1,\dots x_4] \prod_j P[S_j|x_j]
\]
where $P[x_1,\dots x_4]$ is the probability of ancestral data $x_j$ for tree $S_j$. What the inconsistency results suggest here is that the probabilities of the tip data $P[S_j|x_j]$ can be very large for ancestral character states that are quite different from the actual ancestral data. These would be the only character states giving rise to a relatively large likelihood contribution. In effect, in the presence of inconsistent estimation, the likelihood would be $P[x_1,\dots x_4]$, the four taxon likelihood for the tree $12|34$ but using wrong data $x_1,\dots x_4$ for the tips. 

Our results open an interesting direction for future research.
An immediate avenue is extending these results beyond star trees.
Since likelihood-based methods are no-longer the undisputed best options when edge lengths are unknown, one long-term goal is to develop a better alternative approach.
On the other hand, the inconsistency of likelihood-based methods under non-symmetric models may come from the fact that the number of unknown parameters grows with the number of species. A similar phenomenon has been shown in other statistical settings and is sometimes referred to as the Neyman-Scott problem after \citet{neyman1948inconsistency} who showed that estimation of a measurement error variance $\sigma^2$ can be inconsistent when a small number $m$ of repeated measures, $X_{i1},\dots,X_{im}\sim N(\mu_i,\sigma^2)$ are obtained for a large number $i=1,\dots n$ of objects. 
We hypothesize that if edge lengths are generated from a common distribution, then likelihood-based methods will be consistent. Motivation for this approach comes in part from the fact that treating the $\mu_i$ as random in the Neyman-Scott problem can lead to consistent estimation of $\sigma^2$ \citep{kiefer1956consistency}.
The setting for ancestral reconstruction is complicated by a wide range of dependence structures however and so this question remains open.

\section*{Acknowledgement}
We would like to thank Dr. Vu Dinh for the fruitful discussions about the topic.
LSTH was supported by startup funds from Dalhousie University, the Canada Research Chairs program, the
NSERC Discovery Grant RGPIN-2018-05447, and the NSERC Discovery Launch Supplement DGECR-2018-00181.
ES was supported by a Discovery Grant awarded by the Natural Sciences and Engineering Research Council of Canada. 

\appendix

\section{Technical details}

In this section, we provide detailed proofs of our results.
First, we prove the following useful property of Bernoulli expectations.

\begin{lemma}[A property of Bernoulli expectations]
Suppose that $X_1,\dots,X_M$ are independent and $X_i\sim \mbox{Bernoulli}(p_i)$.
For any $j \in \{ 1, 2, \ldots, M\}$, we define
$$x_{-}=[x_1,\dots,x_{j-1},0,x_{j+1},\dots,x_M]$$ where $x \in \{0,1\}^M$.
Let $g(x)$ be any function such that $g(x) - g(x_{-}) \geq 0$ with strict inequality for at least one $x$ that arises with positive probability. 
Then
\[
E_{p}[g(X)]-E_{p'}[g(X)]>0, \quad \forall p> p'
\]
where $E_p[g(X)]$ is the expected value when $X = (X_1, X_2, \ldots, X_M)$ and $p_j = p$.
\label{lem:bern}
\end{lemma}

\begin{proof}

\begin{align*}
E[g(X)] 
&= \sum_{x|x_j=1} g(x) P(X_j=1) \prod_{i \mid i\neq j} P(X_i=x_i) 
   + \sum_{x|x_j=0} g(x) P(X_j= 0) \prod_{i \mid i\neq j} P(X_i=x_i) \\
&= \sum_{x|x_j=1} [g(x)-g(x_{-})] P(X_j= 1) \prod_{i \mid i\neq j} P(X_i=x_i) 
   + \sum_{x|x_j=0} g(x) \prod_{i \mid i\neq j} P(X_i=x_i) 
\end{align*}

Therefore,
\[ E_{p}[g(x)]-E_{p'}[g(x)]
    = \sum_{x|x_j\ge 1} [g(x)-g(x_{-})] 
   \{P_{p}(X_j= 1) - P_{p'}(X_j=1) \} 
       \prod_{i \mid i\neq j} P(X_i=x_i) \ge 0 \]
with strict inequality if $g(x)-g(x_{-})>0$ and $\prod_{i \mid i\neq j} P(X_i=x_i)>0$ for at least one $x$.
\end{proof}

\noindent
Now, we are ready to prove our Lemmas and Theorems in the paper.

\begin{proof}[Proof of Lemma \ref{lem:ex}]

Since $(s_k)_{k=1}^n \in (0,1)$, we have $0 < E(s_1) < + \infty$.
By the Strong Law of Large Numbers, $\bar s$ converges to $E(s_1)$ almost surely.
Thus, Assumption \ref{assump:lims} holds.

\end{proof}

\begin{proof}[Proof of Lemma \ref{lem:MLE}]

Let $g(\mb y)=\log P_{\bs \rho \mb y}(\hat s(\bs \rho, \mb y))$. 
Then,
\[\mbox{Var}[\ell(\bs \rho)/n] = n^{-2} \sum_{\mb y} \mbox{Var}(n_{\mb y}) g(\mb y)^2
            + n^{-2}\sum_{\mb y} \sum_{\mb y'| \mb y'\neq \mb y} \mbox{Cov}(n_{\mb y},n_{\mb y'})g(\mb y) g(\mb y').\]
Multinomial calculations give that 
$\mbox{Var}(n_{\mb y})=\sum_k P_{\bs \rho^* \mb y}(s_k)[1 - P_{\bs \rho^* \mb y}(s_k)]$ and 
$\mbox{Cov}(n_{\mb y},n_{\mb y'})= - \sum_k P_{\bs \rho^* \mb y}(s_k)P_{\bs \rho^* \mb y'}(s_k)$. 
Thus
\[
\begin{split} \mbox{Var}[\ell(\bs \rho)/n] &= n^{-2} \sum_{\mb y} \sum_k P_{\bs \rho^* \mb y}(s_k) g(\mb y)^2
    - n^{-2} \sum_{\mb y} \sum_k [P_{\bs \rho^* \mb y}(s_k)]^2 g(\mb y)^2 \\
&~~~~~~~~~~~- n^{-2} \sum_{\mb y} \sum_{\mb y'| \mb y'\neq \mb y} \sum_k P_{\bs \rho^* \mb y}(s_k) P_{\bs \rho^* \mb y'}(s_k) g(\mb y) g(\mb y')\\
&= n^{-1} \sum_{\mb y} \bar P_{\bs \rho^* \mb y} g(\mb y)^2
     - n^{-2} \sum_{\mb y} \sum_{\mb y'} \sum_k P_{\bs \rho^* \mb y}(s_k) P_{\bs \rho^* \mb y'}(s_k) g(\mb y) g(\mb y')\\
&= n^{-1} \sum_{\mb y} \bar P_{\bs \rho^* \mb y} [\log P_{\bs \rho \mb y}(\hat s(\bs \rho, \mb y))]^2
     - n^{-2} \sum_k \left [ \sum_{\mb y} P_{\bs \rho^* \mb y}(s_k) g(\mb y) \right ]^2 \\
&\le  n^{-1} \sum_{\mb y} \bar P_{\bs \rho^* \mb y} [\log P_{\bs \rho \mb y}(\hat s(\bs \rho, \mb y))]^2 \to 0.
\end{split}
\] 
Note that $P_{\bs \rho \mb y}(\hat s(\bs \rho, \mb y)) > 0$ because it is the maximized likelihood. Since $\mbox{Var}[\ell(\bs \rho)/n]\rightarrow 0$, then for each $\bs \rho$, $\ell(\bs \rho)/n\rightarrow_p \lim_n E[\ell(\bs \rho)/n]=e(\bs \rho)$. 
Since there are a finite collection of possible $\bs \rho$, then for any $\epsilon>0$, with probability converging to 1, $|\ell(\bs \rho)/n-e(\bs \rho)|<\epsilon$ for all $\bs \rho$. 
Let $\bs \rho_M$ be a maximizer of $e(\bs \rho)$ and let $3\epsilon=\min_{\bs \rho \not\in \mc{H}} \{e(\bs \rho_M) - e(\bs \rho)\}>0$. 
Then $|\ell(\bs \rho)/n-e(\bs \rho)|<\epsilon$ for all $\bs \rho$ implies 
\[
l(\bs \rho_M)/n-l(\bs \rho)/n> e(\bs \rho_M)-\epsilon -(e(\bs \rho)+\epsilon) = e(\bs \rho_M)-e(\bs \rho)-2\epsilon\ge 3\epsilon - 2\epsilon> 0
\]
for all $\bs \rho \not \in \mc{H}$. 
In conclusion, $\hat{\bs \rho} \in \mc{H}$, with probability converging to 1.

\end{proof}

\begin{proof}[Proof of Theorem \ref{thm:EmB}]

Let $\bs \rho_M=\bs \rho_M(\mb y)$ be the element of $\mc{H}$ giving the largest $\ell(\bs \rho)$ for fixed $\mb y$. The empirical Bayes estimator 
$\hat{ \bs \rho}_B$ is also the maximizer of 
\[
\frac{ P_\pi(\bs \rho) \prod_{k=1}^n P_{\bs \rho \mb y_k}(\hat s_k) }
     { P_\pi(\bs \rho_M) \prod_{k=1}^n P_{\bs \rho_M \mb y_k}(\hat s_k) }
=
\frac{ P_\pi(\bs \rho) \prod_{k=1}^n P_{\bs \rho \mb y_k} (\hat s_k)}
     { P_\pi(\bs \rho_M) \exp[\ell(\bs \rho_M)] }
\bigg/
\frac{ P_\pi(\bs \rho_M) \prod_{k=1}^n P_{\bs \rho_M \mb y_k}(\hat s_k) }
     { P_\pi(\bs \rho_M) \exp[\ell(\bs \rho_M)] }
=: T_1/T_2.
\]
Since this function equals 1 when $\bs \rho = \bs \rho_M$, it suffices to show that it converges to 0 for $\bs \rho \notin \mc{H}$. 
Since $\hat s(\bs \rho, \mb y_k)$ is the maximizer of $P_{\bs \rho \mb y_k}(s)$ and since $\lim_n \{\ell(\bs \rho_M)/n - \ell(\bs \rho) / n\}>0$, 
\begin{align*}
T_1 &\le 
\frac{ P_\pi(\bs \rho)  \prod_{k=1}^n P_{\bs \rho \mb y_k}(\hat s(\bs \rho, \mb y_k)) }
     { P_\pi(\bs \rho_M) \exp[\ell(\bs \rho_M)] }
=
\frac{ P_\pi(\bs \rho) \exp[\ell(\bs \rho)] }
     { P_\pi(\bs \rho_M) \exp[\ell(\bs \rho_M)] } \\
&  =
\frac{ P_\pi(\bs \rho) }{ P_\pi(\bs \rho_M) } \exp \left \{-n \left [\frac{\ell(\bs \rho_M)}{n} - \frac{\ell(\bs \rho)}{n} \right ] \right \}
\rightarrow 0.
\end{align*}
On the other hand,
\[
T_2 = 
\frac{ P_\pi(\bs\rho_M) \prod_{k=1}^n P_{\bs \rho_M \mb y_k}(\hat s_k) }
     { P_\pi(\bs \rho_M) \exp[\ell(\bs \rho_M)] }
= 
\sum_{\bs \rho} \frac{ P_\pi(\bs \rho) \prod_{k=1}^n P_{\bs \rho \mb y_k}(\hat s_k) }
               { P_\pi(\bs \rho_M) \exp[\ell(\bs \rho_M)] }
- 
\sum_{\bs \rho \neq \bs \rho_M} \frac{ P_\pi(\bs \rho) \prod_{k=1}^n P_{\bs \rho \mb y_k}(\hat s_k) }
                           { P_\pi(\bs \rho_M) \exp[\ell(\bs \rho_M)] }.
\]
Since $\hat{\mb s}$ maximizes the first sum and since $\hat s(\bs \rho, \mb y_k)$ is the maximizer of $P_{\bs \rho \mb y_k}(s)$,
\[\begin{split}
T_2 &\ge
\sum_{\bs \rho} 
  \frac{ P_\pi(\bs \rho) 
             \prod_{k=1}^n P_{\bs \rho \mb y_k}(\hat s(\bs \rho_M, \mb y_k)) }
       { P_\pi(\bs \rho_M) \exp[\ell(\bs \rho_M)]  }
-\sum_{\bs \rho \neq \bs \rho_M} 
   \frac{ P_\pi(\bs \rho) 
             \prod_{k=1}^n P_{\bs \rho \mb y_k}(\hat s(\bs \rho, \mb y_k)) }
        { P_\pi(\bs \rho_M) \exp[\ell(\bs \rho_M)] } \\
&=
1 + \sum_{\bs \rho \neq \rho_M} 
     \frac{ P_\pi(\bs \rho) }{ P_\pi(\bs \rho_M) } 
     \bigg\{1 - \frac{\exp[\ell(\bs \rho)]}{\exp[\ell(\bs \rho_M)]}\bigg\}
\end{split}\]
Because $\bs \rho_M$ maximizes $\ell(\bs \rho)$ over $\bs \rho\in \mc{H}$, then 
$1 - \exp[\ell(\bs \rho)]/\exp[\ell(\bs \rho_M)]\ge 0$ for $\bs \rho\in \mc{H}$. Thus
\[
T_2 \ge 
1 + \sum_{\bs \rho \notin \mc{H}}
\frac{ P_\pi(\bs \rho) }{ P_\pi(\bs \rho_M) }
  \bigg[1 - 
\exp \left \{-n \left [\frac{\ell(\bs \rho_M)}{n} - \frac{\ell(\bs \rho)}{n} \right ] \right \}\bigg]
\]
Since $\lim_n \{\ell(\bs \rho_M)/n - \ell(\bs \rho) / n\}>0,$ we obtain that $\liminf T_2\ge 1$.
Hence, $T_1/T_2 \to 0$.

\end{proof}

\begin{proof}[Proof of Lemma \ref{lem:MLEedge}]

Taking derivatives of both sides of \eqref{eq:prop1}, we obtain
\begin{align}
\label{eq:prop5}
  \frac{\partial}{\partial s} \log P_{\bs \rho \mb y}(s)&= \sum_{i}n_{\bs \rho \mb y}(i,i)  \frac{1-\pi_{i}}{\pi_{i}+ (1-\pi_{i})s} - \frac{N - \sum_{i} n_{\bs \rho \mb y}(i,i)}{1-s} \\
\nonumber
&=  \sum_{i}n_{\bs \rho \mb y}(i,i)  \frac{(1-\pi_{i})(1-s) + \pi_{i} + (1-\pi_{i})s}{[\pi_{i}+ (1-\pi_{i})s](1-s)} - \frac{N}{1-s}\\
\label{eq:prop2}
&= (1-s)^{-1} \left \{\sum_{i}  \frac{n_{\bs \rho \mb y}(i,i)}{\pi_{i}+ (1-\pi_{i})s} - N \right \}.
\end{align}
Taking derivatives of \eqref{eq:prop5},
\[
\frac{\partial^2}{\partial s^2} \log P_{\bs \rho \mb y}(s)= -\sum_{i}n_{\bs \rho \mb y}(i,i)  \frac{(1-\pi_{i})^2}{(\pi_{i}+ (1-\pi_{i})s)^2} - \frac{N - \sum_{i} n_{\bs \rho \mb y}(i,i)}{(1-s)^2} < 0.
\]
Thus, $\log P_{\bs \rho \mb y}(s)$ is concave as a function of $s$. 

When $\sum_i n_{\bs \rho \mb y}(i,i) = N$, up to an additive constant, $$\log P_{\bs \rho \mb y}(s)=\sum_{i} n_{\bs \rho \mb y}(i,i) \log[\pi_{i}+ (1-\pi_{i})s],$$ an increasing function of $s$. So $\log P_{\bs \rho \mb y}(s)$ achieves its maximum at $s = 1$. 
Hence $\hat s(\bs \rho, \mb y) = 1$.

If $\sum_i n_{\bs \rho \mb y}(i,i) < N$, then $N - \sum_{i} n_{\bs \rho \mb y}(i,i)>0$ in the second term of \eqref{eq:prop1}, so $s=1$ gives $\log P_{\bs \rho \mb y}(s)=-\infty.$ 
Thus $\hat s(\bs \rho, \mb y) < 1$ in this case. 
Since $\log P_{\bs \rho \mb y}(s)$ is concave, $\hat s(\bs \rho, \mb y) = 0$ if and only if 
\[
\frac{\partial}{\partial s} \log P_{\bs \rho \mb y}(0) \leq 0.
\]
Substituting $s = 0$ in \eqref{eq:prop2} gives $\hat s(\bs \rho, \mb y) = 0$ if and only if 
\[
\sum_{i}\frac{n_{\bs \rho \mb y}(i,i)}{\pi_{i}} \le N.
\] 

If $\sum_{i}n_{\bs \rho \mb y}(i,i)/\pi_{i} > N$ and $\sum_i n_{\bs \rho \mb y}(i,i) < N$, then $\hat s(\bs \rho, \mb y)$ is obtained by setting \eqref{eq:prop2} to $0$ which gives
\[
  \sum_{i}  \frac{n_{\bs \rho \mb y}(i,i)}{\pi_{i}+ (1-\pi_{i}) \hat s(\bs \rho, \mb y)} = N.
\]

\end{proof}

\begin{proof}[Proof of Theorem \ref{thm:symmetric}]

Note that $\sum_{i} n_{\bs \rho \mb y}(i,i)=\sum_{l=1}^N x_l$, here $x_l=1$ if $y_l=\rho_l$ and 0 otherwise. For the symmetric model, (\ref{eq:prop1}) gives that
\[
\log P_{\bs \rho \mb y}(\hat s)
   = \sum_{l} x_l \cdot \log[1/c + (1-1/c) \hat s] + (N-\sum_l x_l ) \log(1-\hat s) - \sum_l x_l \log(1/c) + C(\mb y)
\]
It follows that $\hat s = \hat s(x)$ is a function of $(\bs \rho, \mb y)$ through $x$ alone. As a function of $x$, $\log P_{\bs \rho \mb y}(\hat s)-C(\mb y)$ is 
\begin{align*}
g(x) &= F(\hat s(x),x)\\
 &=  \sum_{l} x_l \log[1/c + (1-1/c) \hat s(x)] + (N-\sum_l x_l) \log[1-\hat s(x)] - \sum_l x_l \log(1/c) .
\end{align*}
Thus
\begin{align*}
\frac{\partial}{\partial x_j} g(x) 
&= \frac{\partial}{\partial s} F(\hat s(x),x) \cdot 
\frac{\partial\hat s}{\partial x_j} 
+ \frac{\partial}{\partial x_j} F(\hat s(x),x)\\
&=\frac{\partial\hat s}{\partial x_j} \cdot \frac{\partial}{\partial s} \log P_{\rho y}(\hat s) + \log[1/c + (1-1/c) \hat s]- \log(1 - \hat s) - \log[1/c].
\end{align*}
If $0<\hat s<1$, then $\frac{\partial}{\partial s} \log P_{\rho y}(\hat s)=0$. 
so 
\[\frac{\partial}{\partial x_j} g(x) = \log[1 + (c - 1) \hat s ] - \log(1-\hat s) > 0.\]
If $\hat s=0$, 
\[\frac{\partial}{\partial x_j} g(x)= \frac{\partial\hat s}{\partial x_j} \cdot \frac{\partial}{\partial s} \log P_{\rho y}(0). \]
From Lemma \ref{lem:MLEedge},  $\hat s =0$ whenever $\sum_i n_{\bs \rho \mb y}(i,i)/(1/c)=c \sum_l x_l\le N$.
Thus $\frac{\partial\hat s}{\partial x_j}=0$ for $\hat s =0$ and $c\sum_l x_l< N$.
If $\hat s =0$ and $c\sum_l x_l=\sum_i n_{\bs \rho \mb y}(i,i)/(1/c) = N$, then, by Equation \eqref{eq:prop2}, we deduce that $\frac{\partial}{\partial s} \log P_{\rho y}(0)=0$. 
Therefore, $\frac{\partial}{\partial x_j} g(x) =0$ when $\hat s=0$.
Thus, for $0\le \hat s < 1$, $\frac{\partial}{\partial x_j} g(x)\ge 0$ with strict inequality for any $x$ such that $0< \hat s<1$. 
If $\hat s=1$ then $\sum_i n_{\bs \rho \mb y}(i,i) = N$ or each of the $y_l=\rho_l$. In this case $x_j=1$ and cannot be increased. In conclusion, $g(x)-g(x_{-})\ge 0,$ where $x_{-}=[x_1,\dots,x_{j-1},0,x_{j+1},\dots,x_c],$ with strict inequality for all $x$ with $\hat s>0$.

By Equation \eqref{eq:prop1}, an alternative expression for $e(\bs \rho)$ is given by
\[e(\bs \rho) = E[l(\rho)/n] 
= \frac{1}{n} \sum_{k=1}^n E[\log P_{\bs \rho \mb Y_k}(\hat s(\bs \rho,\mb Y_k))] 
= \frac{1}{n} \sum_{k=1}^n E[g(X^{(k)})] + \frac{1}{n} \sum_{k=1}^n E[C(\mb Y _k)].
\]
where $E[g(X^{(k)})]$ is the expected value with respect to the sequence observed at the $k$th leaf.
Here, $X^{(k)} = (X_1, X_2, \ldots, X_N)$ where, independently, $X_i\sim \mbox{Bernoulli}(P_{\rho_i^* \rho_i}(s_k))$.


Consider a $\bs \rho$ with $\rho_j\neq \rho_j^*.$ We denote $h_{\bs \rho}(s_k)=E[g(X^{(k)})]$ and define $\tilde{\bs \rho}$ as $\tilde \rho_i = \rho_i$ for $i \ne j$ and $\tilde \rho_j = \rho_j^*$. Since, for the symmetric model, $P_{\rho_j^* \rho_j^*}(s_k)> P_{\rho_j^* \rho_j}(s_k),$ and $g(x) - g(x_{-})\ge 0$ with strict inequality for all $x$ with $\hat s>0$, Lemma \ref{lem:bern} (the property of Bernoulli expectations) gives that $h_{\tilde{\bs \rho}}(s_k) > h_{\bs \rho}(s_k)$.

For any $U \in (0,1]$, let $s_U = \argmin_{s \in [U, 1]} h_{\tilde{\bs \rho}}(s) - h_{\bs \rho}(s)$.
Denoting the proportion of $s_k \ge U$ as $p_U$,
\[
\bar s \le \sum_{k \mid s_k \ge U} \frac{1}{n}  + \frac{1}{n} \sum_{k \mid s_k < U} U \le p_U + U.
\]
Since $0 < \lim \bar s$, choosing $U = \lim \bar s / 2$, gives that $\lim\inf p_U \ge \lim \bar s/2$ for large $n$.
Thus, for sufficiently large $n$,
\[
\begin{split}
e(\tilde{ \bs \rho})-e(\bs \rho) &=
 \lim \frac{1}{n} \sum_{k} h_{\tilde{\bs \rho}}(s_k) - h_{\bs \rho}(s_k)\\
 &\ge \lim  \frac{1}{n}  \sum_{k \mid s_k \ge U} h_{\tilde{\bs \rho}}(s_k) - h_{\bs \rho}(s_k) \\
& \ge [ h_{\tilde{\bs \rho}}(s_U) - h_{\bs \rho}(s_U) ] \lim\inf p_U \ge  \frac{[ h_{\tilde{\bs \rho}}(s_U) - h_{\bs \rho}(s_U) ] \lim \bar s}{2} > 0. 
\end{split}
\]

In other words, $e(\bs \rho)$ can always be increased by changing a $\rho_j$ that is not equal to the true ancestral state, if such $\rho_j$ exists. Thus $\bs \rho_M = \bs \rho^*$.

\end{proof}

\begin{proof}[Proof of Theorem \ref{thm:zone}]

It is sufficient for us to focus on $\bs \rho=(a,\dots,a)$.
Let $\hat p_a=n_{\bs \rho \mb y}(a, a)/N$ and $\hat s = \hat s(\bs \rho, \mb y)$. 
When $1 > \hat p_a > \pi_a$, by Lemma \ref{lem:MLEedge}, we have
\begin{equation}
\label{eq:prop3a}
\frac{n_{\bs \rho \mb y}(a, a)}{\pi_{a}+ (1-\pi_{a}) \hat s} = N
\Longleftrightarrow 
\hat p_a = \pi_a+ (1-\pi_a) \hat s
\end{equation}
which has solution $\hat s=(\hat p_a- \pi_a)/(1-\pi_a)$. 
Substituting in Equation \eqref{eq:prop1} gives that when $1 > \hat p_a > \pi_a$,
\begin{eqnarray}
\nonumber
\frac{\log P_{\bs \rho \mb y}(\hat s)}{N} &= & \hat p_a
\log \left \{\pi_a + 
       (1-\pi_a) \cdot \frac{\hat p_a- \pi_a}{1-\pi_a} \right \}  
 +[1-\hat p_a] \log\left\{1-\frac{\hat p_a- \pi_a}{1-\pi_a}\right\} \\
\nonumber
 && ~~~~~~~~-  \hat p_a\log(\pi_a) 
+ C(\mb y)/N\\
\label{eq:Elnl1}
&=& \hat p_a \log\left\{\frac{\hat p_a}{\pi_a} \right\}+
   [1-\hat p_a] \log\left\{\frac{1-\hat p_a}{1-\pi_a}\right\}
   + C(\mb y)/N
\end{eqnarray}
In the case that $\hat p_a=1$, $\hat s=1$, which when substituted in Equation \eqref{eq:prop1} gives Equation \eqref{eq:Elnl1} where we adopt the convention that $0\cdot \log(0)=0$. In the case that $\hat p_a \leq \pi_a$, $\log P_{\bs \rho \mb y}(\hat s)= C(\mb y)$.
Thus, for the $k$th sequence,
\[
E \left [ \frac{\log P_{\bs \rho \mb y_k}(\hat s_k)}{N} \right ] =
E\left[\left\{~\hat p_a \log\left[\frac{\hat p_a}{\pi_{a}} \right]+
   (1-\hat p_a) \log\left[\frac{1-\hat p_a}{1-\pi_{a}}\right]~\right\}
   ~I\big\{\hat p_a > \pi_a\big\}\right] +E[C(\mb y)/N],
\]
where $I\{\cdot\}$ is the indicator function. 
For the $k$th sequence, $N\hat p_a \sim\mbox{binomial}(N,P_{ra}(s_k)).$ Let 
$$w(p,\pi)=E\left[\left\{~\hat p \log\left[\frac{\hat p}{\pi} \right]+
   (1-\hat p) \log\left[\frac{1-\hat p}{1-\pi}\right]~\right\}
   ~I\big\{\hat p > \pi\big\}\right],$$
where $N\hat p \sim \mbox{binomial}(N,p)$. 
Since $e(\bs \rho)$ is the limiting average value of 
$E[ \log P_{\bs \rho \mb y_k}(\hat s_k)]$,
\[
\frac{e(\bs \rho) - e(\bs \rho^*)}{N}
= \lim_n \frac{1}{n} \sum_{k=1}^n \{w(P_{ra}(s_k),\pi_a)-w(P_{rr}(s_k),\pi_r)\}.
\] 
If $s_k=0$ all $k$, then $w(P_{ra}(s_k),\pi_a)=v(\pi_a)$ and $w(P_{rr}(s_k),\pi_r)=v(\pi_r)$, so 
\[
\frac{e(\bs \rho) - e(\bs \rho^*)}{N} = v(\pi_a) - v(\pi_r)
\] 
If this value is positive, then by continuity, $e(\bs \rho) - e(\bs \rho^*)$ will be positive for choices of $\mb s$ for which $\lim \bar s$ is sufficiently small.
\end{proof}

\begin{proof}[Proof of Theorem \ref{thm:distinct}]

By Equation \eqref{eq:dchar-1}
\[
E[V(j,l)] = \bar s [I\{\rho_l^* = j\} 
                         - \sum_{v\neq l} I\{\rho_v^* = j\}/(k-1)].
\]
Since there are at least two distinct ancestral character states, 
\[ E[V(\rho_l^*,l)]\ge \bar s \left [1- \frac{N-2}{N-1} \right ] = \frac{\bar s}{N-1} \to \frac{\lim \bar s}{N-1} > 0.\]
On the other hand, for $j\neq {\rho_l^*},$ $E[V(j,l)]\le 0$.
We have
\begin{equation}
\label{eq:dchar-1a}
P[\bs \rho^{(D)}\neq \bs \rho^*; \mb s]
 = P \left [\bigcup_{l=1}^N \{\rho_l^{(D)}\neq \rho_{l}^*\}; \mb s \right ]
  \le \sum_{l=1}^N P[\rho_l^{(D)}\neq \rho_{l}^*; \mb s].
\end{equation}
For any $\epsilon_n( \mb s)$, 
\[\begin{split}
P[\rho_l^{(D)}= \rho_{l}^*; \mb s]
  &= P\left [ \bigcap_{j\neq \rho_{l}^*} \{V(\rho_{l}^*,l) > V(j,l)\}; \mb s \right ]\\
  &\ge P \left [ \bigcap_{j\neq \rho_{l}^*} \{V(j,l) \le \epsilon_n( \mb s )\}
             \bigcap \{V(\rho_{l}^*,l) > \epsilon_n( \mb s )\}; \mb s \right ]\\
  &= 1 - P \left [\bigcup_{j\neq \rho_{l}^*} \{V(j,l) > \epsilon_n(\mb s )\} 
                \bigcup \{V(\rho_{l}^*,l) \le \epsilon_n( \mb s)\} ; \mb s \right]\\
  &\ge 1 - P[V(\rho_{l}^*,l) \le \epsilon_n(\mb s); \mb s] - 
              \sum_{j\neq \rho_{l}^*} P[V(j,l)>\epsilon_n(\mb s); \mb s].
\end{split}\]
So
\begin{equation}
\label{eq:dchar-1b}
P[\rho_l^{(D)} \ne \rho_{l}^*;\mb s]
 = 1- P[\rho_l^{(D)}= \rho_{l}^*;\mb s]
 \le P[V(\rho_{l}^*,l) \le \epsilon_n(\mb s);\mb s] + 
              \sum_{j\neq \rho_{l}^*} P[V(j,l)>\epsilon_n(\mb s);\mb s].
\end{equation}
Now
\[V(j,l) = \frac{1}{n} \sum_{k=1}^n \left [ I\{y_{kl}=j\} - \frac{1}{N-1}\sum_{v\neq l}I\{y_{kv} = j\} \right ]  
        = \frac{1}{n} \sum_{k=1}^n V_k\]
where the $V_k \in [-1,1]$ are independent. 
Thus, using Hoeffding's Inequality,
for $E[V(\rho_{l}^*,l)]-\epsilon_n>0$, we have
\begin{eqnarray}
 P[V(\rho_{l}^*,l)\le \epsilon_n(\mb s);\mb s]
\nonumber
     &=& P\big \{ -V(\rho_{l}^*,l)-E[-V(\rho_{l}^*,l)]\ge E[V(\rho_{l}^*,l)] - \epsilon_n(\mb s); \mb s \big \}\\
\label{eq:dchar-2}
    &\le& \exp[- n \{E[V(\rho_{l}^*,l)] - \epsilon_n(\mb s)\}^2].
\end{eqnarray}
Consider $\epsilon_n(\mb s)= \bar s/[2(N-1)]>0.$ 
Since $E[V(\rho_{l}^*,l)]\ge \bar s/(N-1),$ then $$E[V(\rho_{l}^*,l)]-\epsilon_n(t)\ge \frac{\bar s_n(t)}{2(N-1)}>0.$$
Substituting in (\ref{eq:dchar-2})
\begin{equation}
\label{eq:dchar-3}
P[V(\rho_{l}^*,l)\le \epsilon_n(\mb s);\mb s] \le 
           \exp\{- n \bar s^2/[4(N-1)^2]\}
\end{equation} 
For $j\neq \rho_{l}^*$, $E[V(j,l)] \leq 0$, so 
\begin{eqnarray} 
P[V(j,l) > \epsilon_n(\mb s);\mb s]
\nonumber
   &\le& P[V(j,l)-E[V(j,l)] > \epsilon_n(\mb s);\mb s]\\
\label{eq:dchar-4}
   &\le& \exp[- n \epsilon_n(\mb s)^2]
    = \exp\{- n \bar s^2/[4(N-1)^2]\}
\end{eqnarray}
Combining (\ref{eq:dchar-1a}), (\ref{eq:dchar-1b}), (\ref{eq:dchar-3}) and (\ref{eq:dchar-4}) gives
\begin{equation}
\label{eq:dchar-5}
P[\bs \rho^{(D)}\neq \bs \rho^*;\mb s] \le c N \exp\{-n \bar s^2/[4(N-1)^2]\}
\le c N \exp\{-n \bar s^2/[4N^2]\}.
\end{equation}
Substituting $N = \sqrt{\bar s n/(2\log n )}$ in (\ref{eq:dchar-5}) gives
\[ P[\bs \rho^{(D)}\neq \bs \rho^*; \mb s] 
 \le c \sqrt{\bar s n/(2\log n )} \cdot \exp[-\log(n)/2]
   = c \sqrt{\frac{\bar s}{2\log n}}. \]
So consistency holds with $N = O(\sqrt{n/\log n}).$

\end{proof}

\begin{proof}[Proof of Lemma \ref{lem:twosites}]

Consider the MLE of edge-lengths when likelihoods are calculated with fixed ancestral states, $(\rho_1,\rho_2)$, at the two sites. For a sequence with character states $(x,y)$ at the two sites, when $(x,y)=(\rho_1,\rho_2)$, the MLE is $\hat s = 1,$ giving log likelihood contribution $\log(1)=0$. When $x\neq \rho_1$ and $y\neq \rho_2$, $\hat s = 0$ and the log likelihood contribution is $\log(\pi_x\pi_y).$
When $x=\rho_1$ but $y\neq\rho_2$, the log likelihood contribution for the sequence is
\begin{equation}
\label{eq:ts-a}
\log[\{\pi_{\rho_1}+(1-\pi_{\rho_1})\hat s\}\{1- \hat s\}\pi_{y}]
  = \log(\pi_{\rho_1}+(1-\pi_{\rho_1})\hat s) + \log(1-\hat s) + \log(\pi_y)
\end{equation}
where $\hat s$ is the maximizer of $F(s)=\log[(pi_{\rho_1} + (1-\pi_{\rho_1})s)+\log[1-s].$ 
Here 
\[
F'(s)=\frac{1-\pi_{\rho_1}}{\pi_{\rho_1} + (1-\pi_{\rho_1})s} - \frac{1}{1-s}.
\] 
Note that $ F'(s) = 0$ has a unique solution in $(-\infty, 1]$:
\[
s_0 = \frac{1-2\pi_{\rho_1}}{2(1-\pi_{\rho_1})}
\]
and $F(s) > 0 \Leftrightarrow s < s_0$.
If $\pi_{\rho_1} \ge 1/2$, then $F'(s) < 0$ for all $s \in [0,1]$.
So $\hat s=0$ and the log likelihood contribution (\ref{eq:ts-a}) is $\log(\pi_{\rho_1}) + \log(\pi_y).$ 
If $\pi_{\rho_1}< 1/2$, $F(s)$ achieves maximum at $s_0$.
Thus, $\hat s  = s_0$.
Substituting in (\ref{eq:ts-a}) gives log likelihood contribution
\begin{equation}
\label{eq:ts-c}
 \log[1/2]-\log[2(1-\pi_{\rho_1})]+\log[\pi_y]=-\log[4(1-\pi_{\rho_1})]+\log[\pi_y]
\end{equation}
By symmetry, if $x\neq\rho_1$ but $y=\rho_2$, the log likelihood contributions are $\log(\pi_{\rho_2}) + \log[\pi_x]$ when $\pi_{\rho_2}\ge 1/2$ and $-\log[4(1-\pi_{\rho_2})]+\log[\pi_x]$ when $\pi_{\rho_2}<1/2.$

Let $n_{y_1 y_2}$ be the number times we observe the sequence $(y_1, y_2)$, we define
\[
n^{(1)}_{y_1} = \sum_{y_2}{n_{y_1 y_2}}, \quad ,\text{and} \quad n^{(2)}_{y_2} = \sum_{y_1}{n_{y_1 y_2}}.
\]
Substituting the log likelihood contributions gives
\begin{eqnarray}
\nonumber
\ell(\rho_1,\rho_2) &=& \sum_{x\neq\rho_1}\sum_{y\neq\rho_2} n_{xy} \log[\pi_x \pi_y]  
     + \sum_{y\neq\rho_2} n_{\rho_1y} \log[\pi_y] + \sum_{x\neq\rho_1} n_{x \rho_2} \log[\pi_x]  \\
\nonumber
    && +\sum_{y\neq\rho_2} n_{\rho_1y} \{I\{\pi_{\rho_1}\ge 1/2\} \log[\pi_{\rho_1}] - I\{\pi_{\rho_1}< 1/2\} \log[4(1-\pi_{\rho_1})]\} \\
  \label{eq:ts-1}
    &&  +\sum_{x\neq\rho_1} n_{\rho_1x} \{I\{\pi_{\rho_2}\ge 1/2\} \log[\pi_{\rho_2}] - I\{\pi_{\rho_2}< 1/2\} \log[4(1-\pi_{\rho_2})]\}.
\end{eqnarray}
Combining terms involving $\pi_x$ gives the sum
\begin{eqnarray}
\nonumber
 \sum_{x\neq\rho_1}\sum_{y\neq\rho_2} n_{xy} \log[\pi_x] + \sum_{x\neq\rho_1} n_{x\rho_2} \log[\pi_x] 
   &=& \sum_{x\neq\rho_1} \sum_{y} n_{xy}  \log[\pi_x] \\
  \nonumber
   &=& \sum_{x\neq\rho_1} n^{(1)}_{x}\log[\pi_x] \\
     \label{eq:ts-2}
   &=& \sum_{x} n^{(1)}_{x}\log[\pi_x] - n^{(1)}_{\rho_1}\log[\pi_{\rho_1}].
\end{eqnarray}
Similarly, 
\begin{equation}
  \label{eq:ts-3}
  \sum_{x\neq\rho_1}\sum_{y\neq\rho_2} n_{xy} \log[\pi_y] + \sum_{y\neq\rho_2} n_{\rho_1y} \log[\pi_y] 
   = \sum_{y} n^{(2)}_{y}\log[\pi_y] - n^{(2)}_{\rho_2}\log[\pi_{\rho_2}] .
\end{equation}
Substituting (\ref{eq:ts-2}) and (\ref{eq:ts-3}) in (\ref{eq:ts-1}) gives
\begin{eqnarray}
\nonumber
\ell(\rho_1,\rho_2) &=& \sum_{x} n^{(1)}_{x}\log[\pi_x]+\sum_{y} n^{(2)}_{y}\log[\pi_y] - n^{(1)}_{\rho_1}\log[\pi_{\rho_1}] - n^{(2)}_{\rho_2}\log[\pi_{\rho_2}] \\
\nonumber
  && +[n^{(1)}_{\rho_1}-n_{\rho_1\rho_2}] \{I\{\pi_{\rho_1}\ge 1/2\} \log[\pi_{\rho_1}] - I\{\pi_{\rho_1}< 1/2\} \log[4(1-\pi_{\rho_1})]\} \\
  \label{eq:ts-4}
  && +[n^{(2)}_{\rho_2}-n_{\rho_1\rho_2}] \{I\{\pi_{\rho_2}\ge 1/2\} \log[\pi_{\rho_2}] - I\{\pi_{\rho_2}< 1/2\} \log[4(1-\pi_{\rho_2})]\}.
\end{eqnarray}
Note that
\[
I\{\pi\ge 1/2\} \log[\pi] - I\{\pi< 1/2\} \log[4(1-\pi)]= \log[\pi] - I\{\pi< 1/2\}\log[4\pi(1-\pi)].
\] 
Substituting in (\ref{eq:ts-4}) with 
$f(\pi)=I\{\pi< 1/2\}\log[4\pi(1-\pi)],$ 
\begin{eqnarray}
\nonumber
\ell(\rho_1,\rho_2) &=& \sum_{x} n^{(1)}_{x}\log[\pi_x]+\sum_{y} n^{(2)}_{y}\log[\pi_y] - n^{(1)}_{\rho_1}\log[\pi_{\rho_1}] - n^{(2)}_{\rho_2}\log[\pi_{\rho_2}] \\
\nonumber
&& ~~~~+[n^{(1)}_{\rho_1}-n_{\rho_1\rho_2}] \{\log[\pi_{\rho_1}] - f(\pi_{\rho_1})\}
 + [n^{(2)}_{\rho_2}-n_{\rho_1\rho_2}] \{\log[\pi_{\rho_2}]  - f(\pi_{\rho_2})\}\\
  \nonumber
&=&  \sum_{x} n^{(1)}_{x}\log[\pi_x]+\sum_{y} n^{(2)}_{y}\log[\pi_y] - n^{(1)}_{\rho_1} f(\pi_{\rho_1}) - n^{(2)}_{\rho_2} f(\pi_{\rho_2}) \\
  \label{eq:ts-5}
&& ~~~~   + n_{\rho_1\rho_2}\{f(\pi_{\rho_1})+ f(\pi_{\rho_2})
                         -\log[\pi_{\rho_1}\pi_{\rho_2}]\}.
\end{eqnarray}
The expectations of $n^{(1)}_{\rho_1}/n,$ $n^{(2)}_{\rho_2}/n$ and $n_{\rho_1\rho_2}/n$ are $\bar P_{\rho_1^*\rho_1}$, $\bar P_{\rho_2^*\rho_2 }$ and  
$\bar P_{(\rho_1^*, \rho_2^*) (\rho_1, \rho_2)}$ respectively. 
So the expected maximized log likelihood is
\[
E[l(\rho_1,\rho_2)/n] = C_0 - \bar P_{\rho_1^*\rho_1} f(\pi_{\rho_1}) - \bar P_{\rho_2^*\rho_2 } f(\pi_{\rho_2}) + \bar P_{(\rho_1^*, \rho_2^*) (\rho_1, \rho_2)} \{f(\pi_{\rho_1})+ f(\pi_{\rho_2})
                         -\log[\pi_{\rho_1}\pi_{\rho_2}]\}
\]
where 
\[
C_0 = \sum_{x} \bar P_{\rho_1^* x } \log[\pi_x]+\sum_{y} \bar P_{\rho_2^* y }\log[\pi_y]
\]  is independent of $(\rho_1,\rho_2)$.

\end{proof}

\bibliographystyle{chicago}
\bibliography{ms}
\end{document}